\documentclass[journal,onecolumn,11pt]{IEEEtran}

\usepackage[margin=1in]{geometry}
\usepackage{amsmath,amsthm,amssymb,bm}
\usepackage{arydshln}
\usepackage{hyperref,url,xcolor}
\usepackage{multirow}
\newcommand{\josh}[1]{{\color{orange!50!black}[Josh: #1]}}
\newcommand{\gopi}[1]{{\color{blue!50!black}[Gopi: #1]}}

\newcommand{\F}{\mathbb F}
\newcommand{\PF}{\mathbb {PF}}
\newcommand{\cF}{\mathcal F}

\newcommand{\poly}{\operatorname{poly}}
\newcommand{\ceil}[1]{\lceil #1\rceil}
\newcommand{\rank}{\operatorname{rank}}
\newcommand{\Moore}{\operatorname{Moore}}

\newtheorem{theorem}{Theorem}

\newtheorem{proposition}[theorem]{Proposition}
\newtheorem{corollary}[theorem]{Corollary}

\newtheorem{claim}[theorem]{Claim}
\newtheorem{remark}[theorem]{Remark}

\newtheorem*{question*}{Question}

\newcommand{\change}[1]{#1}

\hypersetup{
	colorlinks=true,
	linkcolor=blue,%
	citecolor=blue,
	linktoc=page
}

\title{Improved Constructions and Lower Bounds for Maximally Recoverable Grid Codes}
\author{Joshua~Brakensiek, Manik~Dhar, and Sivakanth~Gopi
\thanks{J. Brakensiek is with the Department of Electrical Engineering and Computer Sciences, University of California, Berkeley. Email: \texttt{josh.brakensiek@berkeley.edu}. This work is supported in part by DMS-2503280.}%
\thanks{M. Dhar is with the Department of Mathematics at the Massachusetts Institute of Technology. Email: \texttt{dmanik@mit.edu}.}%
\thanks{S. Gopi is with Microsoft Research. Email: \texttt{sigopi@microsoft.com}.}}
\date{}

\begin{document}
\maketitle

\begin{abstract}
In this paper, we continue the study of Maximally Recoverable (MR) Grid Codes initiated by Gopalan et al. [SODA 2017]. More precisely, we study codes over an $m \times n$ grid topology with one parity check per row and column of the grid along with $h \ge 1$ global parity checks. Previous works have largely focused on the setting in which $m = n$, where explicit constructions require field size which is exponential in $n$. Motivated by practical applications, we consider the regime in which $m,h$ are constants and $n$ is growing. In this setting, we provide a number of new explicit constructions whose field size is polynomial in $n$. We further complement these results with new field size lower bounds.
\end{abstract}

\section{Introduction}

Inspired by the data-storage needs of modern distributed systems, Gopalan et al.~\cite{Gopalan2016} initiated the study of maximally recoverable (MR) codes for grid-like topologies. In this model, we think of our data storage topology as an $m \times n$ grid. Each of the $mn$ cells of the grid stores as information an element of $\F_q$. To impose some local redundancy, for each of the $n$ columns (with $m$ symbols each), we impose $a \ge 0$ linear parity checks. Likewise, for each of the $m$ rows (with $n$ symbols each), we impose $b \ge 0$ linear parity checks. Finally, in addition to these row and column checks, we impose $h \ge 0$ more ``global'' parity checks. We call such a code a $(m, n, a, b, h)$ grid code.

Given these topological constraints, we seek to seek to have our $(m, n, a, b, h)$ grid code be resilient to various \emph{erasure errors}. In particular, given a subset $E$ of our $m \times n$ grid, can we recover the symbols stored by $E$ using the various parity checks? In general, the answer depends on the precise linear-algebraic properties of the specific parity check matrix chosen in the design of the code. We say that our code is an MR $(m,n,a,b,h)$ grid code if our code is information-theoretically optimal. In other words, \emph{any} pattern $E$ which is correctable by \emph{some} $(m,n,a,b,h)$ grid code is correctable by an MR $(m,n,a,b,h)$ grid code.

The main question we seek to study concerning MR $(m,n,a,b,h)$ grid codes is as follows:

\begin{question*}
What is the minimum field size $q(m,n,a,b,h)$ for which we can construct an MR $(m,n,a,b,h)$ grid code?
\end{question*}

In this paper, we make a number of new contributions to this question in the regime for which $a=b=1$ and $h\ge 1$, particularly when $m$ and $h$ are constants. Practical instantiations of MR grid codes typically have $m\ll n$. For example the f4 storage system from Meta \cite{MLRH14} uses the $(m=3,n=14,a=1,b=4,h=0)$ grid topology. The reason for choosing $m\ll n$ is that typically the grid rows are placed in different zones or geographically separated datacenters (like in f4) and are supposed to protect against zone or datacenter failures. Since such failures are rare events which require large communication for reconstruction, typically only a small $m$ is chosen. 

\subsection{Literature Overview}

Before we describe our results, we first survey the rich literature on MR $(m,n,a,b,h)$ grid codes. A key motivation for Gopalan et al.'s definition of an MR $(m,n,a,b,h)$ grid code is that the regime in which $b = 0$ (or symmetrically $a=0$) captures the topology of MR locally recoverable codes (LRC) (also known as \emph{partial maximal distance separable (PMDS)} codes). In particular, we think of the $mn$ data symbols as being partitioned into $n$ groups of size $m$, such that there are $a$ local parity checks within each group as well as $h$ additional global parity checks. A vast body of literature \cite{chen2007maximally,gopalan2012locality,huang2012erasure,huang2013pyramid,BlaumPS13,gopalan2014explicit,papailiopoulos2014locally,hu2016new,martinez2019universal,gopi2020maximally,cai2021construction,gopi2022improved,dhar2023construction,martinez2025maximally} has studied the structure of MR LRCs, including a wide variety of constructions, some of which have essentially optimal field size. We defer the reader to a recent papers by Cai et al.~\cite{cai2021construction} and by Dhar and Gopi~\cite{dhar2023construction} for a more thorough overview of the literature.

Another well-studied regime is the case in which $h = 0$. In that case, MR $(m,n,a,b,0)$ grid codes are more commonly referred to as MR $(m,n,a,b)$ tensor codes, because the generator matrix of a $(m,n,a,b)$ tensor code is the tensor (or Kronecker) product of two smaller generator matrices (corresponding to the row and column codes). This topology (although not MR) has been used in practice by Meta~\cite{MLRH14}. The case $a=1$ is the most well-studied, especially in the connection to \emph{higher order MDS codes},  \cite{shivakrishna2018Maximally,kong2021new,bgm2021mds,roth2021higher,holzbaur_correctable_2021,athi2023Structurea,bgm2022,brakensiek2024Generalized} and in particular has been found to be closely linked to the construction of optimal list-decodable codes \cite{roth2021higher,bgm2022,guo2023randomly,alrabiah2023randomly,alrabiah2023ag,brakensiek2025ag,AGGLZ25}. See Brakensiek et al.~\cite{brakensiek2023improved} for an overview of field size bounds for higher order MDS codes. In general for $a, b \ge 2$ (and $h = 0$) the situation is poorly understood. In fact, we even lack an efficient description of which patterns $E$ can be recovered by an MR $(m,n,2,2)$ grid code. However, Brakensiek et al.~\cite{brakensiek2024rigidity} recently proved that any such characterization is equivalent to characterizing the which graphs are \emph{bipartitely rigid}~\cite{kalai2015Bipartite}. See Brakensiek et al.~\cite{brakensiek2024rigidity} and references therein for a more detailed discussion.

Unlike the two previous large bodies of work, much less study has been put into MR $(m,n,a,b,h)$ grid codes when all three parameters $a,b,h$ are at least one. Gopalan et al.~\cite{Gopalan2016} gave a description of the correctable patterns of an $(m,n,1,1,h)$ grid code, although the situation was greatly clarified by Holzbaur et al.~\cite{holzbaur_correctable_2021} who showed that any pattern $E$ correctable by an $(m,n,a,b,h)$ grid code can be formed by taking a pattern correctable by an $(m,n,a,b,0)$ grid (a.k.a. tensor) code and erasing at most $h$ additional symbols. Furthermore, they show that if an MR $(m,n,a,b)$ tensor code exists over field size $q$, then an MR $(m,n,a,b,h)$ grid code exists over field size $q^{(m-a)(n-b)}$. In other words, $q(m,n,a,b,h) \le q(m,n,a,b,0)^{(m-a)(n-b)}$. When $a=b=h=1$ an upper bound of $q(m,n,1,1,1) \le 2^{3\max(m,n)}$ is known via a construction due to Kane et al.~\cite{kane2019independence}. For $a=b=1$ and $h \ge 2$, the best explicit construction due to Gopalan et al.~\cite{gopalan2014explicit,Gopalan2016} obtains $q(n,n,1,1,h) \le n^{O(n)}$.

Few results are known in terms of lower bounds.
Gopalan et al.~\cite{Gopalan2016} showed that $q(n,n,1,1,1) \ge n^{\Omega(\log n)}$. A breakthrough by Kane et al.~\cite{kane2019independence} substantially improved this lower bound to $q(n,n,1,1,1) \ge 2^{n/2-2}$ using techniques from representation theory. Coregliano and Jeronimo~\cite{coregliano2022tighter} improved this lower bound to $q(n,n,1,1,1) \ge 2^{0.97n-O(1)}$, albeit under the assumption that the characteristic of the field is two. For $h\ge 2$, Gopalan et al.~\cite{Gopalan2016} (as stated by Kane et al.~\cite{kane2019independence}) states that
\begin{align}
q(m,n,a,b,h) \ge q(\min(m-a+1,n-b+1),\min(m-a+1,n-b+1),1,1,1)^{1/h}.\label{eq:lb}
\end{align} Combining this with the best lower bounds for $h=1$, results in $q(n,n,1,1,h)\ge 2^{\Omega(n/h)}.$ In the regime considered in this paper, i.e., $m,h=O(1)$ and $n$ large, the best known lower bound is only $2^{\Omega(m/h)}$ which doesn't grow with $n$. It is also surprising that existing lower bounds get smaller with increasing $h$, whereas intuitively it is harder to construct codes with larger $h$.

\subsection{Results}

\begin{table}
\begin{center}
\def\arraystretch{1.4}
\begin{tabular}{|c|cccl|}
\hline
& $m$ & $h$ & Bound on $q(m, n, a=1, b=1, h)$ & Reference\\\hline
\multirow{5}{4em}{Upper Bounds}&$\ge 2$ & $\ge 1$ &  $n^{O(n)}$ & Gopalan et al.~\cite{gopalan2014explicit,Gopalan2016}\\
&$\ge 2$ & $1$ & $2^{3n}$ & Kane et al.~\cite{kane2019independence}\\
&$\ge 2$ & $\ge 1$ & $2^{(m-1)(n-1)}$ & Holzbaur et al.~\cite{holzbaur_correctable_2021}\\
\cdashline{2-5}
&$\ge 2$ & $1$ & $n^{m-1}$ &  Theorem~\ref{thm:h1m-1}\\
&$\ge 2$ & $\ge 1$ & $2 (2mn)^{m+h-2}$ & Theorem~\ref{thm:BCH-zero}\\
&$3$ & $1$ & $n \cdot 2^{O(\sqrt{\log n})}$ & Corollary~\ref{cor:m=3}\\
\hline
\multirow{5}{4em}{Lower Bounds}& $\ge 2$ & $\ge 1$ & $\Omega_{m,h}(1)$ & \cite{Gopalan2016,kane2019independence,coregliano2022tighter}\\
\cdashline{2-5}
&$\ge 2$ & $1$ & $n$ & Corollary~\ref{cor:mge2lb}\\
&$\ge 2$ & $\ge 1$ & $n^{\lceil h/2\rceil}$ & Theorem~\ref{thm:gab-lower} (\emph{Gabidulin construction only})\\
&$\ge (h-1)^2$ & $\ge 3$ & $\Omega_h(n^{h-1})$ & Theorem~\ref{thm:lb-general}\\
&$\ge 3$ & $\ge 1$ & $q(m-2, n-h,a=1,b=1,h=1)$ & Theorem~\ref{thm:monotone-h}\\
&$\ge 2$ & $\ge 2$ & $\Omega_{m,h}(n^{\min(\lfloor \sqrt{m-2}\rfloor, h-1)})$ & Corollary~\ref{cor:combine}\\
\hline
\end{tabular}
\end{center}
\caption{A comparison of our bounds on $q(m, n, a=1, b=1, h)$ with prior work. To simplify the presentation, we assume that $m$ and $h$ are constants relative to $n$ and that $n$ is a power of two.}\label{table:results}
\end{table}

As previously mentioned, we investigate in this paper the setting in which $a=b=1$ and $h \ge 1$. Our primary focus is on the regime for which $m$ and $h$ are constants relative to $n$ which is growing. This setting has essentially no results in the literature. In particular, existing bounds only imply a constant lower bound on the field size, and existing constructions have field size which are exponential in $n$.

In this paper, we provide a number of new constructions and lower bounds to show that the correct dependence on $n$ is $\poly(n)$.  We highlight the following results--see Table~\ref{table:results} for a more comprehensive overview. First, we give a simple construction for $h=1$, with $\poly(n)$ field size.

\begin{theorem}[Informal version of \change{Theorem~\ref{thm:h1m-1}}]
$q(m,n,1,1,1) \le n^{m-1}$.
\end{theorem}

The key idea is encoding the binary representation of the coordinates of each symbol. By combining BCH codes and Gabidulin codes, we are able to extend this construction to any $h \ge 2$, beating the previous best bound of $2^{(m-1)(n-1)}$ due to Holzbaur et al.~\cite{holzbaur_correctable_2021}.

\begin{theorem}[Informal version of Theorem~\ref{thm:BCH-zero}]
$q(m,n,1,1,h) \le O(mn)^{m+h-2}$.
\end{theorem}

We further give a field size lower bound when $m$ is a sufficiently large function of $h$, beating the previously best lower bound of constant in this regime.

\begin{theorem}[Informal version of Theorem~\ref{thm:lb-general}]
Assume $m \ge (h-1)^2$, then $q(m,n,1,1,h) \ge \Omega(n^{h-1}).$
\end{theorem}

This result is proved by combining the characterization of correctable patterns due to Holzbaur et al.~\cite{holzbaur_correctable_2021} with lower bounds techniques developed by Gopi et al.~\cite{gopi2018Maximally} for MR LRCs. Though the upper and lower bounds we prove for the $m,h=O(1)$ regime are polynomial in $n$ whose exponent is growing with $m,h$, we leave open the question of figuring exactly the correct exponent of $n$ as a function of $m,h$. %

We also improve the lower bound (\ref{eq:lb}) by showing that $q(m,n,1,1,h) \ge q(m-2,n-h,1,1,1)$, which shows that the constructing codes for $h\ge 2$ is at least as hard as for $h=1$ (see Theorem~\ref{thm:monotone-h}). We remark that some of our techniques also work in the regime in which $m$ is growing as a function of $n$. These extensions are discussed in the Appendix\change{, including the following result:}

\change{
\begin{theorem}[Informal version of Theorem~\ref{thm:h1-bootstrap}]
Assume that $m \le n$ are powers of two, then $q(m, n, 1, 1, 1) \le 8(8n/m)^{m-1}$.
\end{theorem}
}

\section{Problem Setting}\label{sec:problem}
To work toward these results, we now formally define the MR grid code problem. We define $[m] := \{1, 2, \hdots, m\}$. In our $m \times n$ grid code, we always assume $m \le n$. We describe a $(m,n,a=1,b=1,h)$ grid code by its parity check matrix $H \in \F_q^{(m+n+h-1) \times mn}$. \change{For this parity check matrix to be sensible, we always assume that $m+n+h-1 \le mn$.} Note that any $(i,j) \in [m]\times [n]$ corresponds to the column indexed $n(i-1)+j$ of $H$. In other words, column $\ell$ of $H$ corresponds to the grid square $(\lceil \ell / n\rceil, (\ell \!\!\mod n) + 1)$. We now describe $H$ as follows, see Figure~\ref{fig:ex} for an example.

\begin{itemize}
\item The first $m$ rows of $H$ correspond to the row parity checks. In particular, for each $i \in [m]$ and $\ell \in [mn]$, $H_{i,\ell} = 1$ if $\ell \in \{(i-1)n+1, (i-1)n+2, \hdots, in\}$ and $0$ otherwise. %
\item The next $n-1$ rows of $H$ correspond to the column parity checks. Note that the $n$th column parity check is implied by the sum of the $m$ row parity checks minus the sum of the other $n-1$ column parity checks. For all $j \in [n]$ and $\ell \in [mn]$, we have that $H_{j+m,\ell} = 1$ if $\ell \in \{i, j+n, \hdots, j + (m-1)n\}$ and $0$ otherwise.
\item The final $h$ rows of $H$ correspond to the global parity checks. We currently do not specify these parity checks as they depend on the specifics of each construction, but we introduce the following notation. For all $i \in [m], j \in [n], k \in [h]$, we let $c^{i,j}_k$ be alternative notation for the entry $H_{m+n+k-1, n(i-1)+j}$. In particular, $c^{i,j} := (c^{i,j}_1, \hdots, c^{i,j}_k) \in \F_q^h$ is the $(i,j)$th parity check vector.
\end{itemize}

\begin{figure}
\def\arraystretch{1.2}
\setcounter{MaxMatrixCols}{20}
\[
    \begin{pmatrix}
    1 & 1 & 1 & 1 & 1 & & & & & & & & & &\\
    & & & & & 1 & 1 & 1 & 1 & 1 & & & &\\
    & & & & & & & & & & 1 & 1 & 1 & 1 & 1\\
    1 & & & & & 1 & & & & & 1 & & & &\\
    & 1 & & & & & 1 & & & & & 1 & & &\\
    & & 1 & & & & & 1 & & & & & 1 & &\\
    & & & 1 & & & & & 1 & & & & & 1 &\\
    c^{1,1}_1 & c^{1,2}_1 & c^{1,3}_1 & c^{1,4}_1 & c^{1,5}_1 & c^{2,1}_1 & c^{2,2}_1 & c^{2,3}_1 & c^{2,4}_1 & c^{2,5}_1 & c^{3,1}_1 & c^{3,2}_1 & c^{3,3}_1 & c^{3,4}_1 & c^{3,5}_1\\
    c^{1,1}_2 & c^{1,2}_2 & c^{1,3}_2 & c^{1,4}_2 & c^{1,5}_2 & c^{2,1}_2 & c^{2,2}_2 & c^{2,3}_2 & c^{2,4}_2 & c^{2,5}_2 & c^{3,1}_2 & c^{3,2}_2 & c^{3,3}_2 & c^{3,4}_2 & c^{3,5}_2\\
    \end{pmatrix}
\]
\caption{Parity check matrix for $(m=3, n=5, a=1, b=1, h=2)$.}
\label{fig:ex}
\end{figure}

We let $C$ denote the set of all $c \in \F^{mn}$ such that $Hc = 0$. Given a pattern $E \subseteq [m] \times [n]$, we let $H|_{E}$ denote the restriction of $H$ to the columns indexed by $\{n(i-1)+j \mid (i,j) \in E\}$. We say that $E$ is \emph{correctable} \change{for any $c \in C$ we have that $c|_{\bar{E}} = 0$ iff $c = 0$, where $\bar{E} = [m] \times [n] \setminus E$}. \change{Equivalently, $E$ is \emph{not} correctable if there exists a nonzero $c \in C$ such that $c_{i,j} \neq 0$ only if $(i,j) \in E$. Since $H c = 0$ for all $c \in C$, the existence of such a codeword is equivalent to there being a linear dependence among the columns of $H|_E$. In other words, we have the following proposition.}

\change{
\begin{proposition}
$E \subseteq [m] \times [n]$ is correctable if and only if $\rank H|_{E} = |E|$.
\end{proposition}
}
\subsection{Correctable Patterns}

Holzbaur et al.~\cite{holzbaur_correctable_2021} give a full characterization of the correctable patterns of a MR $(m, n, a=1, b=1, h)$ grid code. First, we recall the characterization of the $h = 0$ case by Gopalan et al.~\cite{Gopalan2016}. In particular, $E \subseteq [m] \times [n]$ is correctable in a MR $(m, n, a=1, b=1)$ tensor code if and only if $E$ is \emph{acyclic}. That is, if we interpret $E$ as the edges of a bipartite graph with \emph{left vertices} $[m]$ and \emph{right vertices} $[n]$, then the result graph lacks a simple cycle. Here, for any $k\ge 2$, we define a simple cycle of length $2k$ to be a sequence of edges
\[
(i_1, j_1),  (i_2, j_1), (i_2, j_2), (i_3, j_2), \hdots, (i_k, j_k), (i_1, j_k) \in [m] \times [n]
\]
where $i_1, \hdots, i_k \in [m]$ and $j_1, \hdots, j_k \in [n]$ are pairwise distinct. 

Holzbaur et al.~\cite{holzbaur_correctable_2021} give the following characterization of the correctable patterns of $(m, n, a=1, b=1, h)$ grid codes. A similar characterization appears in Gopalan et al.~\cite{Gopalan2016}.

\begin{theorem}[\cite{holzbaur_correctable_2021}]\label{thm:a=b=1-pat}
A pattern $E \subseteq [m] \times [n]$ is correctable by a MR $(m, n, a=1, b=1, h)$ grid code if and only if there exists $E' \subseteq E$ of size at least $|E|-h$ such that $E'$ is acyclic.
\end{theorem}

We note they prove an analogous result for arbitrary $a,b \ge 1$, but we only use the $a=b=1$ setting in this paper. We say our parity check matrix $H \in \F_q^{(m+n+h-1) \times mn}$ is MR $(m,n,a=1,b=1,h)$ if $\rank H|_{E} = |E|$ for every pattern $E$ which is an \change{acyclic} graph plus at most $h$ edges. This criterion can be a bit cumbersome to check directly, so we next give a somewhat simpler criterion for checking if a parity check matrix is MR.

\subsection{Cycle Sum Independence}

Assume we are given parity check vectors $c^{i,j} \in \F^h_q$ for all $(i,j) \in [m] \times [n]$ and a corresponding parity check matrix $H \in \F_q^{(m+n+h-1) \times mn}$. Given a cycle $C := \{(i_1, j_1),  (i_2, j_1), \hdots, (i_k, j_k), (i_1, j_k)\} \subseteq [m] \times [n]$, we define $C$'s \emph{cycle sum} to be 
\begin{align}
    \Sigma_H(C) := \sum_{\ell=1}^k c^{i_\ell, j_\ell} - c^{i_{\ell+1}, j_\ell} \in \F_q^h\label{eq:cycle-sum}
\end{align}
where we define $i_{k+1} := i_1$. Note that if $q$ is odd, then the above definition depends on the orientation of the cycle, although these different sums are equal up to a scalar (i.e., projectively equivalent). %

\change{Consider} a spanning tree $T \subseteq [m] \times [n]$ (i.e., an acyclic graph for which every vertex is incident with at least one edge). \change{For} every $(i,j) \in [m] \times [n] \setminus T$ \change{note that there is a unique path connecting $i$ and $j$. Since $i$ and $j$ are on opposite sides of a bipartite graph, this path has an odd number of edges, say $2k-1$. Denote the vertices of this path by $i = p_1, p_2, \hdots, p_{2k} = j$. Since this path from $i$ to $j$ within $T$ is unique, there is exactly one cycle which uses edges from $T \cup \{(i,j)\}$. Precisely, this cycle is
\[
  C_{T,(i,j)} := \{(i,j),(i=p_1,p_2),(p_3,p_2),(p_3,p_4)\hdots, (p_{2k-1},p_{2k-2}),(p_{2k-1},j=p_{2k})\}.
\]} 
\change{For example, if $m=n=2$, and $T = \{(1,1),(1,2),(2,2)\}$, then $(2,1) \in [m] \times [n] \setminus T$ and $C_{T,(2,1)} = \{(2,1),(2,2),(1,2),(1,1)\}$.} 

\begin{proposition}\label{prop:cycle-ind}
$H \in \F_q^{(m+n+h-1) \times mn}$ is MR$(m,n,a=1,b=1,h)$ if and only if for all spanning trees $T \subseteq [m] \times [n]$ and distinct edges $(i_1, j_1), \hdots, (i_h, j_h) \in [m] \times [n] \setminus T$, we have that $\Sigma_H(C_{T,(i_1,j_1)}), \hdots, \Sigma_H(C_{T,(i_h,j_h)}) \in \F_q^h$ are linearly independent.
\end{proposition}
Note that in the case $h=1$, we recover the condition that ``no cycle sum equals zero'' found by Gopalan et al.~\cite{Gopalan2016}.
\begin{proof}
By Theorem~\ref{thm:a=b=1-pat}, we have that $H$ is MR if and only if $\rank H|_{E} = \change{|E|}$ for every $E$ \change{which is the union of an acyclic graph and at most $h$ additional edges. If there is a strictly larger $E' \supsetneq E$ of this form, then note that $\rank H|_{E'} = \change{|E'|}$ implies $\rank H|_{E} = |E|.$} In other words, we may assume that $E = T \cup \{(i_1, j_1), \hdots, (i_h, j_h)\}$ where $T$ is a spanning tree\footnote{\change{If $E$ does not contain a spanning tree, then viewed as a subgraph of $[m] \times [n]$, $E$ has multiple connected components. Adding any edge bridging these connected components will not change the property that $E$ is an acyclic graph plus at most $h$ additional edges.}} \change{and $(i_1,j_1), \hdots, (i_h,j_h) \not\in T$ are distinct}. \change{Therefore, $|E| = m+n-1+h$.} Since $H$ has $m+n+h-1$ rows, it is equivalent to check that $\det H|_{E} \neq 0$.

We now perform a series of column operators on $H$. For each $\ell \in [h]$, let $C_\ell$ be shorthand for $C_{T,(i_\ell,j_\ell)}$. We seek to replace the $(i_\ell, j_\ell)$th column of $H$ with (essentially) $\Sigma_H(C_\ell)$. To do this, assume that
\[
    C_\ell = \{(i'_1, j'_1), (i'_2, j'_1), \hdots, (i'_k, j'_k), (i'_1, j'_k)\},
\]
where $(i'_1, j'_1) = (i_\ell, j_\ell)$ and all other edges of $C_\ell$ are in $T$. We then replace that column $H|_{(i_\ell, j_\ell)}$ with
\[
    \hat{H}_\ell := \sum_{\ell' = 1}^k H|_{(i_{\ell'},j_{\ell'})} - H|_{(i_{\ell'+1},j_{\ell'})}.
\]
One can verify that the first $m+n-1$ entries of $\hat{H}_\ell$ are zeros and the last $h$ entries are precisely $\Sigma_H(C_\ell)$. After these column operations, we have that $H|_{E}$ looks like the following block matrix.

\begin{align*}
\left[H|_{T} \middle| \hat{H}_1, \hdots, \hat{H}_h\right] &= \left[\begin{array}{c|c}
A_T & 0 \\\hline
B_T & \Sigma_H(C_1) \cdots \Sigma_H(C_h)
\end{array}\right],
\end{align*}
where $A_T$ is a $\change{(}n+m-1\change{)} \times \change{(}n+m-1\change{)}$ matrix and $B_T$ is a $h \times \change{(}n+m-1\change{)}$ matrix. Note that $\Sigma_H(C_1) \cdots \Sigma_H(C_h)$ form a $h \times h$ matrix, so we have that
\[
    \det(H|_{E}) = \det(A_T) \det(\Sigma_H(C_1) \cdots \Sigma_H(C_h)).
\]
Recall we are trying to prove that $\det(H|_{E}) \neq 0$ if and only if $\det(\Sigma_H(C_1) \cdots \Sigma_H(C_h))$. Thus, it suffices to prove that $\det(A_T) \neq 0$ unconditionally. To see why, note that the first $m+n-1$ rows of $H$ are precisely the parity check matrix of a MR $(m,n,a=1,b=1,h=0)$ tensor code. Furthermore, \change{the erasure pattern corresponding to the edges of a spanning tree $T$} is recoverable by such an MR tensor code~\change{\cite{Gopalan2016}.} Thus, $\rank A_T = |T|$. Equivalently, $\det(A_T) \neq 0$ since $A_T$ is square, as desired. Therefore, $H|_{E}$ has full column rank if and only if $\Sigma_H(C_{T,(i_1,j_1)}), \hdots, \Sigma_H(C_{T,(i_h,j_h)}) \in \F_q^h$ are linearly independent.
\end{proof}

For use in Section~\ref{sec:lower-bounds}, we prove a variant of Proposition~\ref{prop:cycle-ind} in which the cycles do not need to come from a common spanning tree. However, the implication is now only in one direction.

\begin{proposition}\label{prop:cycle-ind-h'}
Let $H \in \F_q^{(m+n+h-1) \times mn}$ be MR$(m,n,a=1,b=1,h)$. For $h' \le h$. Let $C_1, \hdots, C_{h'} \subseteq [m] \times [n]$ be cycles with the following properties:
\begin{itemize}
\item[(1)] There exists an acyclic subgraph $T \subseteq C_1, \hdots, C_{h'}$ such that $|(C_1 \cup \cdots \cup C_{h'}) \setminus T| \le h$.
\item[(2)] For all $\ell \in [h']$, there is $(i_\ell, j_\ell) \in C_\ell$ which does not appear in $C_{\ell'}$ for all $\ell' \in [h'] \setminus \{\ell\}$.
\end{itemize}
Then, $\Sigma_H(C_1), \hdots, \Sigma_H(C_{h'}) \in \F_q^h$ are linearly independent.
\end{proposition}

\begin{proof}
Let $E' = C_1 \cup \cdots \cup C_{h'}$. By property (1), we have that $|E' \setminus T|\le h$ and $T$ is acyclic. Thus, we have that $E'$ is correctable by Theorem~\ref{thm:a=b=1-pat}. Therefore, $\rank H|_{E'} = |\change{E'}|$.

Like in Proposition~\ref{prop:cycle-ind}, for each $\ell \in [h']$, we can perform a sequence of column operations which replaces the $(i_\ell, j_\ell)$th column of $\rank H|_{E'}$ with $m+n-1$ zeros followed by $\Sigma_H(C_\ell)$. By property (2), we have that $(i_\ell, j_\ell) \not\in C_{\ell'}$ for all $\ell \neq \ell' \in [h]$. Therefore, we can perform the column operations for all $\ell \in [h]$ simultaneously. After these column operations, $H|_{E'}$ must still have full column rank. This can only be possible if $\Sigma_H(C_1), \hdots, \Sigma_H(C_{h'}) \in \F_q^h$ are linearly independent, as desired.
\end{proof}

\section{Constructions}

We now work on describing our constructions for MR $(m,n,a=1,b=1,h)$ codes.

\subsection{The Gabidulin Construction}

The proof of Theorem~\ref{thm:a=b=1-pat} is constructive in the following sense. 
The use of Gabidulin codes from \cite{holzbaur_correctable_2021} converts any MR $(m,n,a,b,h=0)$ code over $\F_q$ into an MR $(m,n,a,b,h)$ code over a field of size $q^{(m-a)(n-b)}$. \change{Since an MR $(m,n,a=1,b=1,h=0)$ code exists over $\F_2$ (see \cite{Gopalan2016}), we obtain} a $2^{(m-1)(n-1)}$ field size construction for an MR $(m,n,a=1,b=1,h)$ grid code. In paper, we construct Gabidulin codes over much smaller field sizes with the same MR property.

Assume that $q = p^d$, where $p$ is a prime and $d \ge h$. Given $\alpha_1, \hdots, \alpha_h \in \F_q$, we define the \emph{Moore matrix} generated by $\alpha_1, \hdots, \alpha_h$ to be
\[
\Moore(\alpha_1, \hdots, \alpha_h) := \begin{pmatrix}
    \alpha_1^{p^0} & \alpha_2^{p^0} & \cdots & \alpha_h^{p^0}\\
    \alpha_1^{p^1} & \alpha_2^{p^1} & \cdots & \alpha_h^{p^1}\\
    \vdots & \vdots & \ddots & \vdots\\
    \alpha_1^{p^{h-1}} & \alpha_2^{p^{h-1}} & \cdots & \alpha_h^{p^{h-1}}
    \end{pmatrix}.
\]
We make use of the following fact concerning the non-singularity of the Moore-matrix.
\begin{proposition}[e.g., \cite{gabidulin2021rank}]\label{prop:Moore}
$\det \Moore(\alpha_1, \hdots, \alpha_h) \neq 0$ if and only if for all $\alpha_1, \hdots, \alpha_h$ are linearly independent over $\F_p$. In other words, for all $b_1, \hdots, b_h \in \F_p$, we have that $b_1 \alpha_1 + b_2\alpha_2 + \cdots + b_h\alpha_h \neq 0$.
\end{proposition}

We say that our parity check matrix $H \in \F^{m+n+h-1 \times mn}_q$ is a \emph{Gabidulin construction} if for all $(i,j ) \in [m] \times [n]$, there exists $\gamma\change{(}i,j\change{)} \in \F_q$ such that $c^{i,j} = (\gamma\change{(}i,j\change{)}, \gamma\change{(}i,j\change{)}^p, \hdots, \gamma\change{(}i,j\change{)}^{p^{h-1}}).$

Note that the map $x \mapsto x^p$ is a \emph{Frobenius automorphism} of $\F_q$ into itself. In particular, $(x+y)^p = x^p + y^p$ and $(cx)^p = cx^p$ for all \change{$x,y \in \F_q$ and $c \in \F_p$, where $\F_p$ is the base subfield of $\F_q$}. As such, for a given cycle $C = \{(i_1, j_1),  (i_2, j_1), \hdots, (i_k, j_k), (i_1, j_k)\}$ $\subseteq [m] \times [n]$, if $H$ is a Gabidulin construction, then $\Sigma_H(C) = \sum_{\ell=1}^k c^{i_\ell, j_\ell} - c^{i_{\ell+1}, j_\ell} =$ $(\gamma\change{(}C\change{)}, \gamma\change{(}C\change{)}^p,$ $\hdots, \gamma\change{(}C\change{)}^{p^{h-1}}),$
where
\[
    \gamma\change{(}C\change{)} := \sum_{\ell=1}^k \gamma\change{(}i_\ell, j_\ell\change{)} - \gamma\change{(}i_{\ell+1}, j_\ell\change{)}.
\]
Thus, as an immediate application of Proposition~\ref{prop:Moore}, Proposition~\ref{prop:cycle-ind} can be recast as follows when $H$ is a Gabidulin construction.

\begin{corollary}\label{cor:cycle-ind}
Assume that $H$ is a Gabidulin construction based on $\gamma : [m] \times [n] \to \F_q$. Then, $H$ is MR$(m,n,a=1,b=1,h)$ if and only if for all spanning trees $T \subseteq [m] \times [n]$ and distinct edges $(i_1, j_1), \hdots, (i_h, j_h) \in [m] \times [n] \setminus T$, we have that $\gamma(C_{T,(i_1,j_1)}), \hdots, \gamma(C_{T,(i_h,j_h)}) \in \F_q$ are linearly independent over $\F_p$.
\end{corollary}

Recall $q = p^d$. If $p = 2$, then linear independence over $\F_p$ is equivalent to any $h' \le h$ cycle sums not summing in total to $0$. For convenience, we often think of the map $\gamma : [m] \times [n] \to \F_q$ instead as the map $\gamma : [m] \times [n] \to \F_p^d$.

Another important fact we need is that the $h$ cycles we consider in Corollary~\ref{cor:cycle-ind} have an overall small ``footprint.''
\begin{proposition}\label{prop:cycle-union}
Let $T \subseteq [m] \times [n]$ be a spanning tree. For any $(i_1, j_1), \hdots, (i_h,j_h) \in [m] \times [n] \setminus T$, we have that
\begin{align}
    \left|\bigcup_{\ell=1}^h C_{T,(i_\ell, j_\ell)}\right| \le 2(m+h-1).\label{eq:cycle-union}
\end{align}
\end{proposition}
\begin{proof}
\change{If $m = 1$, then $T = [m] \times [n]$ so the LHS is (\ref{eq:cycle-union}) is always zero, while $2(m+h-1) = 2h \ge 0$. Thus, we assume $m \ge 2$.}

\change{Say an edge $(i,j) \in T$ is \emph{right-lonely} if $(i,j)$ is the only edge of $T$ containing $j \in [n]$. Let $T' \subseteq T$ be the (maximal) subgraph of $T$ with all right-lonely edges removed. Let $S \subseteq [n]$ be the right vertices incident to $T'$. As such, $T'$ has at most $|S|+m-1$ edges. Thus, in order for $T'$ to lack any further right-lonely edges, we must have that $|S|+m-1 \ge |T'| \ge 2|S|$. Therefore, $m-1\ge |S|$, so $|T'| \le 2(m-1)$.}

We now prove (\ref{eq:cycle-union}). For each $\ell \in [h]$, we claim that $|C_{T,(i_\ell, j_\ell)} \setminus \change{T'}| \le 2$. To see why, $j_\ell$ is connected to two left vertices $i_\ell$ and some $i'_\ell \in [m]$. Since $C_{T, (i_\ell, j_\ell)} \setminus \{(i_\ell,j_\ell)\} \subseteq T$, we have that the path connecting $i_\ell$ and $i'_\ell$ inside $C_{T,(i_\ell,j_\ell)}$ which \change{avoids} $j_\ell$ must \change{have that every right vertex has degree at least two.} \change{Therefore, $C_{T,(i_\ell,j_\ell)} \setminus \{(i_\ell,j_\ell), (i'_\ell, j_\ell) \subseteq T'$}. \change{Hence,} $|C_{T,(i_\ell, j_\ell)} \setminus \change{T'}| \le |\{(i_\ell,j_\ell),(i'_\ell,j_\ell)\}| \le 2$. Thus,
\[
\left|\bigcup_{\ell=1}^h C_{T,(i_\ell, j_\ell)}\right| \le |T_{[m]}| + \sum_{\ell=1}^h |C_{T,(i_\ell, j_\ell)} \setminus T_{[m]}| \le 2m-2 + 2h = 2(m+h-1),
\]
which proves (\ref{eq:cycle-union}). 
\end{proof}
\begin{remark}
Proposition~\ref{prop:cycle-union} is tight for $n \ge m+h-1$. For example, consider $T \subseteq [m] \times [n]$ as follows.
\[
    T = \{(i,i) \mid i \in [m]\} \cup \{(i,i+1) \mid i \in [m-1]\} \cup \{(m,j) \mid j \in \{m+1, \hdots, n\}\}
\]
For each $\ell \in [h]$, let $(i_\ell,j_\ell) = (\change{1}, m+\ell-1)$. It is straightforward to check that (\ref{eq:cycle-union}) is tight in this case.
\end{remark}

\subsection{Binary Encoding Construction}
Now that we have established some general properties of the Gabidulin family of constructions, we now give some specific Gabidulin constructions which are MR grid codes. We start with the case $h=1$.

\begin{theorem}\label{thm:h1m-1}
Assume that $n$ is a power of $2$. Then there is a Gabidulin construction which is MR $(m, n, a=1,b=1,h=1)$ with $q = n^{m-1}$.
\end{theorem}
\begin{proof}
Let $d := \change{(m-1)} \log_2 n$ and let $q = 2^d$. Note that the additive structure of $\F_q$ is isomorphic to $\F^d_2$. For all $(i, j) \in [m] \times [n]$, let $\gamma(i,j) \in \F^d_2$ have the following structure. If $i \le m-1$, then we let $(\gamma(i,j)_{(i-1)\log_2 n + 1}, \hdots, \gamma(i,j)_{i\log_2 n})$ be the binary representation of $j$, with all other coordinates of $\gamma(i,j)$ equal to zero. Otherwise, if $i = m$, we let $\gamma(i,j) = 0$.

Since we are working in characteristic two and $h=1$, it suffices to prove for every simple cycle $C \subseteq [m] \times [n]$ that $\gamma(C) \neq 0$. To see why $\gamma(C) \neq 0$, note that $C$ must have length at least 4. Thus, there are at least two left vertices $i, i' \in [m]$ with nonzero degree in $C$. We may assume without loss of generality that $i \le m-1$. There are exactly two distinct $j, j' \in [n]$ for which $(i,j), (i,j') \in C$. In particular, \change{the coordinates of $\gamma(C)$ indexed by $\{(i-1)\log_2n+1, \hdots, i\log_2 n\}$} are the XOR of the binary representations of $j$ and $j'$. Since $j \neq j'$, this XOR is nonzero, as desired.
\end{proof}

\subsection{BCH Code Construction}

We next describe some more sophisticated constructions for general $h$ which make use of the BCH code~\cite{bose1960class,hocquenghem1959codes}. The columns of the parity check matrix of the BCH code have the following property.

\begin{proposition}[e.g., Exercise~5.9~\cite{guruswami2023essential}]\label{prop:BCH}
For any $D, N \ge 1$, let $d=1+\ceil{(D-1)/2}\ceil{\log_2(N+1)}$. Then, there exists $\alpha_1, \hdots, \alpha_N \in \F_{2}^d$ which are the columns of a parity check matrix of a code with Hamming distance at least $D$. In other words, for all nonempty $S \subseteq [N]$ of size at most $D-1$, $\sum_{i \in S} \alpha_i \neq 0$.
\end{proposition} 

This leads to the following construction for any $h \ge 1$. We remark a similar approach is used by Gopalan et al.~\cite{gopalan2014explicit}.

\begin{theorem}\label{thm:BCH-simple}
Assume $n$ is a power of two and let $2^k$ be the least power of two greater than $m$. Then, there is a Gabidulin construction which is MR $(m, n, a=1,b=1,h)$ with $q = 2(2^k n)^{m+h-1} \le 2(2mn)^{m+h-1}$.
\end{theorem}
\begin{proof}
By Proposition~\ref{prop:BCH} with $D = 2(m+h-1)+1$ and $N = mn$ (so $\log_2 q = 1 + (m+h-1)\lceil \log_2(mn+1)\rceil = 1 + (m+h-1)\log_2(2^k n)$), there exist $\alpha_1, \hdots, \alpha_{mn} \in \F_q$ such that any $2(m+h-1)$ choices of $\alpha_i$ are linearly independent over $\F_2$. We claim that the Gabidulin construction with evaluation points $\gamma(i,j) = \alpha_{n(i-1) + j}$ is MR$(m,n,a=1,b=1,h)$.

By Corollary~\ref{cor:cycle-ind}, it suffices to show that any $h$ cycles $C_1, \hdots, C_h \subseteq [m] \times [n]$ formed by adding $h$ edges to a spanning tree $T$ satisfy $\gamma(C_1), \hdots, \gamma(C_h)$ are $\F_2$-linearly independent. By Proposition~\ref{prop:cycle-union}, we know that $|C_1 \cup \cdots \cup C_h| \le 2(m+h-1)$. Thus, by our choice of BCH code, $\{\gamma(i,j) : (i,j) \in C_1 \cup \cdots \cup C_h\}$ are linearly independent. However, for each $\ell \in [h]$, $C_\ell$ has an edge which does not appear in any other cycle $C_{\ell'}$ for $\ell' \in [h]$. Thus, since $\gamma(C_\ell) = \sum_{(i,j) \in C_\ell} \gamma(i,j)$, we have that $\gamma(C_1), \hdots, \gamma(C_h)$ are linearly independent, as desired.
\end{proof}

When $h=1$, we get an exponent of $m$ rather than $m-1$ like in Theorem~\ref{thm:h1m-1}. We can make Theorem~\ref{thm:BCH-simple} more similar to Theorem~\ref{thm:h1m-1} by adding an extra left vertex whose edges labels are all zero.

\begin{theorem}\label{thm:BCH-zero}
Assume $n$ is a power of two and let $2^k$ be the least power of two greater than $m-1$. Then, there is a Gabidulin construction which is MR $(m, n, a=1,b=1,h)$ with $q = 2(2^k n)^{m+h-2} \le 2(2mn)^{m+h-2}$.
\end{theorem}

\begin{proof}
Let $\gamma : [m-1] \times [n] \to \F_q$ be the edge labeling of Theorem~\ref{thm:BCH-simple} (where the $m$ chosen in that Theorem is $m-1$). In particular, the entries of $\gamma$ come from a BCH code with parameters $D = 2(m+h-2)+1$ and $N = (m-1)n$. We extend this to $\gamma : [m] \times [n] \to \F_q$ by setting $\gamma(m, j) = 0$ for all $j \in [n]$. We claim this construction is MR $(m, n, a=1,b=1,h)$. To see why, consider any $h$ cycles $C_1, \hdots, C_h \subseteq [m] \times [n]$ formed by adding $h$ edges to a spanning tree $T$. If $C_1, \hdots, C_h \subseteq [m-1] \times [n]$, then we know that $\gamma(C_1), \hdots, \gamma(C_h)$ are linearly independent by the previous analysis.

Otherwise, it suffices to prove for any $h' \in [h]$ and any $1 \le \ell_1 < \cdots < \ell_{h'} \le h$ that $\gamma(C_{\ell_1}) + \cdots +\gamma(C_{\ell_{h'}}) \neq 0$. In other words, it suffices to show that $\gamma(C_{\ell_1} \oplus \cdots \oplus C_{\ell_{h'}}) \neq 0$, where $\oplus$ denote the exclusive OR (XOR) of the sets.

By Proposition~\ref{prop:cycle-union}, we know that $|C_1 \cup \cdots \cup C_h| \le 2(m+h-1)$. By assumption, there is at least one $\ell \in [h]$ for which $C_\ell$ is incident with the left vertex $m$. Therefore, at least two edges of $C_1 \cup \cdots \cup C_h$ involve left vertex $[m]$. Thus, $|(C_1 \cup \cdots \cup C_h) \cap [m-1] \times [n]| \le 2(m+h-2)$. In other words,
\[
    \gamma(C_{\ell_1} \oplus \cdots \oplus C_{\ell_{h'}}) = \sum_{(i,j) \in (C_{\ell_1} \oplus \cdots \oplus C_{\ell_{h'}}) \cap [m-1] \times [n]} \gamma(i,j)
\]
has at most $2(m+h-2)$ terms in its sum. 
Since the nonzero value of $\gamma$ come from a BCH code with $D = 2(m+h-2)+1$, it suffices to prove that $(C_{\ell_1} \oplus \cdots \oplus C_{\ell_{h'}}) \cap [m-1] \times [n]$ is nonempty. First, note that $C_{\ell_1} \oplus \cdots \oplus C_{\ell_{h'}}$ is nonempty as each cycle has an edge not appearing in any other cycle. Thus, the only way in which $(C_{\ell_1} \oplus \cdots \oplus C_{\ell_{h'}}) \cap [m-1] \times [n]$ is empty is for $C_{\ell_1} \oplus \cdots \oplus C_{\ell_{h'}} \subseteq \{m\} \times [n]$. In a cycle, every vertex has even degree. Thus, in the XOR of these cycles, every vertex still has even degree. However, there is no nonempty subgraph of $\{m\} \times [n]$ in which every vertex has even degree, a contradiction. Therefore, $\gamma(C_{\ell_1}) + \cdots \gamma(C_{\ell_{h'}}) \neq 0$ so $\gamma(C_1), \hdots, \gamma(C_h)$ are linearly independent, as desired.
\end{proof}

%

%

%

%

%

\subsection{Additive Combinatorics Construction for \texorpdfstring{$m=3$}{m=3} and \texorpdfstring{$h=1$}{h=1}}

We conclude this section with a specialized construction when $m=3$ and $h=1$. Recall Theorem~\ref{thm:h1m-1} says that one can construct such a code over field size $q = n^2$. We now improve this to $q = n^{1+o(1)}$ with the caveat that $q$ is odd. We say that $A \subseteq \F_q$ has no arithmetic progressions of length $3$ (``3-AP free'') if for all distinct $x,y,z \in A$, we have that $x + y \neq 2z$. We first state how 3-AP free sets can be used to construct MR $(m=3, n, a=1, b=1, h=1)$ codes. We remark that 3-AP sets have also been used to construct MR LRCs~\cite{gopi2018Maximally}.

\begin{theorem}\label{thm:3-AP-con}
Let $q$ be an odd prime power and let $A \subseteq \F_q$ be a 3-AP free set of size $n$. Then, one can construct an MR $(m=3, n, a=1, b=1, h=1)$ grid code over $\F_q$.
\end{theorem}
\begin{proof}
Index the 3-AP free set by $A = \{a_1, \hdots, a_n\}$. We consider the Gabidulin construction (which is a general construction for $h=1$) with the map $\gamma : [m] \times [n] \to \F_q$ as follows:
\begin{align*}
\gamma(\change{1}, j) = 0, \gamma(\change{2}, j) = a_j, \gamma(\change{3}, j) = 2a_j \text{ for all $j \in [n]$}.
\end{align*}
We now seek to prove for every simple cycle $C \subseteq [m] \times [n]$ we have that $\gamma(C) \neq 0$. Since $m=3$, $C$ must have $4$ or $6$ edges. If $C$ has four edges, then there are distinct $i_1,i_2 \in [m]$ and $j_1, j_2 \in [n]$ such that $C = \{(i_1,j_1),(i_2,j_1),(i_2,j_2),(i_1,j_2)\}.$ In this case, we have that
\begin{align*}
    \gamma(C) &= \gamma(i_1,j_1) - \gamma(i_2, j_1) + \gamma(i_2, j_2) - \gamma(i_1, j_2)\\
    &= (i_1 - 1)a_{j_1} - (i_2 - 1)a_{j_1} + (i_2 - 1)a_{j_2} - (i_1 - 1)a_{j_2}\\
    &= (i_1 - i_2)(a_{j_1} - a_{j_2}) \neq 0,
\end{align*}
because $i_1 \neq i_2$ and $a_{j_1} \neq \change{a_{j_2}}$. Note this uses the fact that $q$ is odd so that $1, 2, 3$ are distinct in $\F_q$.

If $C$ has six edges, then there are distinct $j_1, j_2, j_3 \in [n]$ such that $C = \{(1,j_1),(2,j_1),(2,j_2),(3,j_2),$ $(3,j_3),(1,j_3)\}.$ We then have that
\begin{align*}
    \gamma(C) &= \gamma(1,j_1)-\gamma(2,j_1)+\gamma(2,j_2)-\gamma(3,j_2)+\gamma(3,j_3)-\gamma(1,j_3)\\
    &= 0 - a_{j_1} + a_{j_2} - 2a_{j_2} + 2a_{j_3} - 0\\
    &= -(a_{j_1} + a_{j_2} - 2a_{j_3}) \neq 0,
\end{align*}
because $A$ is $3$-AP-free, as desired.
\end{proof}

We thus would like to know for a given $n$, what is the minimum field size $q$ needed? If $q$ is prime, the answer is $q = n \cdot 2^{O(\sqrt{\log n})}$~\cite{behrend1946sets,moser1953non,elkin2011improved}, which is also essentially optimal~\cite{kelley2023strong}. See~\cite{gopi2018Maximally} for a more thorough discussion of the literature. Thus, as a corollary of Theorem~\ref{thm:3-AP-con}, we have the following result.

\begin{corollary}\label{cor:m=3}
If $q = n \cdot 2^{O(\sqrt{\log n})}$ is an odd prime, then there is an MR $(m=3, n, a=1, b=1, h=1)$ grid code over $\F_q$.
\end{corollary}

We remark that for odd prime power $q = p^d$, one can also beat $n^2$ based on the best known constructions to the \emph{cap set} problem~\cite{romera2024mathematical}, although such constructions have strictly worse asymptotics than in the prime case~\cite{ellenberg2017large,croot2017progression}.

\section{Field-size Lower Bounds}\label{sec:lower-bounds}

Now that we have established a few different constructions, we now turn to limitations on possible constructions. First, we study the limits of the Gabidulin family of constructions. 

\subsection{\texorpdfstring{$n^{\ceil{h/2}}$}{pow(n,ceil(h/2))} Lower Bound for the Gabidulin Construction}
A simple implication of Corollary~\ref{cor:cycle-ind} gives us a lower bound for the field size of the Gabidulin construction.

\begin{theorem}\label{thm:gab-lower}
    If a Gabidulin construction is used to construct the parity check matrix of an MR$(m,n,a=1,b=1,h)$ code with $m \ge 2$ then the field size must be at least $\binom{n}{\lceil h/2 \rceil}$.
\end{theorem}

Note for $h=1$, the Gabidulin construction is equivalent to a general construction, so this lower bound applies in general in that case.

\begin{corollary}\label{cor:mge2lb}
If $m \ge 2$, then $q(m,n,1,1,1) \ge n$.
\end{corollary}

\begin{proof}[Proof of Theorem~\ref{thm:gab-lower}]
For the lower bound, we assume without loss of generality that $m=2$. Label the two left vertices as $\alpha$ and $\beta$ and label the right vertices as $1,\hdots,n$. If $h$ is odd then we claim that for any $(h+1)/2$ vertices $j_1,\hdots,j_{(h+1)/2} \in [n]$ the sum of $\gamma\change{(}\alpha,j_1\change{)}-\gamma\change{(}\beta,j_1\change{)}+\cdots+\gamma\change{(}\alpha,j_{(h+1)/2}\change{)}-\gamma\change{(}\beta,j_{(h+1)/2}\change{)}$ is distinct for distinct choices of right vertices. We note this claim would then force the field size to be at least $\binom{n}{(h+1)/2}$.

We prove this by contradiction. If two such sums are equal then for two disjoint set of right vertices $i_1,\hdots,i_k$ and $j_1,\hdots,j_k$ we must have
\begin{equation}\label{eq-gabLowSum}
\gamma\change{(}\alpha,i_1\change{)}-\gamma\change{(}\beta,i_1\change{)}+\hdots+\gamma\change{(}\alpha,i_{k}\change{)}-\gamma\change{(}\beta,i_{k}\change{)}-(\gamma\change{(}\alpha,j_1\change{)}-\gamma\change{(}\beta,j_1\change{)}+\hdots+\gamma\change{(}\alpha,j_{k}\change{)}-\gamma\change{(}\beta,j_{k}\change{)})=0\end{equation} for $k\le (h+1)/2$. \change{Let $T$ be the following spanning tree of the subgraph induced by the aforementioned edges:
\[
  T = \{(\alpha,i_1), \hdots, (\alpha,i_k), (\alpha, j_1), \hdots, (\alpha, j_k), (\beta, i_1)\}.
\]
For any vertex $x \in \{i_2, \hdots, i_k, j_1, \hdots, j_k\}$, the edge $(\beta, x)$ forms with $T$ a unique cycle \[C_x := \{(\alpha,i_1), (\beta,i_1),(\beta, x), (\alpha, x)\}.\]
In particular (after making an arbitrary sign choice), we have that
\[
\gamma(C_x) = \gamma(\alpha,i_1) - \gamma(\beta,i_1) + \gamma(\beta, x) - \gamma(\alpha, x).
\]
Since our Gabidulin construction is MR $(n, m, a=1, b=1, h=2k-1)$, we have that $\{\gamma(C_x) : x \in \{i_2, \hdots, i_k, j_1, \hdots, j_k\}\}$ are linearly independent. In particular,
\begin{align*}
0 &\neq \sum_{x \in \{j_1, \hdots, j_k\}} \gamma(C_x) - \sum_{x \in \{i_2, \hdots, i_k\}} \gamma(C_x)\\
&= \gamma(\alpha,i_1) - \gamma(\beta,i_1) + \sum_{x \in \{i_2, \hdots, i_k\}} (\gamma(\alpha, x) - \gamma(\beta, x)) - \sum_{x \in \{j_1, \hdots, j_k\}} (\gamma(\alpha, x) - \gamma(\beta, x))\\
&= 0, & (\text{by \eqref{eq-gabLowSum}})
\end{align*}
}
a contradiction.

A similar argument holds for even $h$ \change{by picking $k=h/2$ and noting that $h \ge 2k-1$}.
\end{proof}

\begin{remark}
To directly prove Corollary~\ref{cor:mge2lb}, note that the set $\{\gamma(1,i)-\gamma(2,i) : i \in [n]\}$ must consist of $n$ distinct elements. Otherwise, if $\gamma(1,i)-\gamma(2,i) = \gamma(1,j) - \gamma(2,j)$ for some $i \neq j$, we have that the cycle $\{(1,i),(2,i),(2,j),(1,j)\}$ has its sum equal to zero, a contradiction. Since $\{\gamma(1,i)-\gamma(2,i) : i \in [n]\}$ can have at most $q$ elements, we must have that $q \ge n$.
\end{remark}

\subsection{General Lower Bounds}

Next, we study field-size lower bounds which apply to any MR grid code. Here we make the assumption that $m$ is sufficiently large compared to $h$.

\begin{theorem}\label{thm:lb-general}
Assume that $h \ge 3$ and $m \ge (h-1)^2$. If there exists an MR $(m, n, a=1, b=1, h)$ grid code over $\F_q$, then $q \ge \Omega(n^{h-1}/h^{2h-1})$.
\end{theorem}

The proof will make use of the following fact due to Gopi et al.~\cite{gopi2018Maximally}. We say that a set of nonzero points $A \subseteq \F_q^h$ are \emph{projectively distinct} if every pair of distinct vectors $a, a' \in A$ are linearly independent. In other words, $A$ is projectively distinct if the elements of $A$ map to distinct points in the projective plane $\PF_q^{h-1}$.

\begin{proposition}[Lemma~3.1~\cite{gopi2018Maximally}]\label{prop:q-ind}
Assume that disjoint sets $A_1, \hdots, A_h \subseteq \PF_q^{h-1}$ are projectively distinct and each have size at least $N$. Further assume that there is no choice of $a_1 \in A_1, \hdots, a_h \in A_h$ which are linearly dependent. Then, $q \ge \frac{N}{h-1}-4 = \Omega(N/h).$
\end{proposition}

In our application, we let $A_1, \hdots, A_h$ each correspond to a family of cycle sums where each cycle has length $2(h-1)$. To ensure projective distinctness, we ensure the $h$ cycles have (mostly) disjoint sets of vertices, leading to our requirement that $m \ge (h-1)^2$.

\begin{proof}[Proof of Theorem~\ref{thm:lb-general}]
Without loss of generality, we assume that $m = (h-1)^2$. Let $I_1, \hdots, I_h \subseteq [m]$ be the following intervals.
\begin{align*}
I_1 &:= \{1, 2, \hdots, h-1\}\\
I_2 &:= \{h-1, h+1, \hdots, 2h-3\}\\
I_3 &:= \{2h-3, 2h, \hdots, 3h-4\}\\
&\vdots\\
I_h &:= \{(h-1)(h-2)+1, \hdots, (h-1)^2\}.
\end{align*}
In particular, each interval has size $\change{h-1}$ and for all $\ell \in [h-1]$, we have that $I_{\ell} \cap I_{\ell+1} = \{\ell(h-2)+1\}$. Furthermore, let $J_1, \hdots, J_{h(h-1)} \subseteq [n]$ be a partition of $[n]$ such that each set has size at least $N := \lfloor \frac{n}{h(h-1)}\rfloor$.

For all $\ell \in [h]$, we define a family $\cF_{\ell}$ of cycles as follows. Let $i_1, \hdots, i_{h-1}$ be an enumeration of the left vertices of $I_\ell$. For any $j_1 \in J_{(\ell-1)(h-1)+1}, j_2 \in J_{(\ell-1)(h-1)+2}, \hdots, j_{h-1} \in J_{\ell (h-1)}$, we define a cycle $C_{\ell, j_1, \hdots, j_{h-1}} \in \cF_{\ell}$ as follows:
\[
    C_{\ell, j_1, \hdots, j_{h-1}} := \{(i_1, j_1), (i_2, j_1), (i_2, j_2), \hdots, (i_{h-1}, j_{h-1}), (i_1, j_{h-1})\}.
\]
For this cycle to be simple, we need at least $|I_\ell| \ge 2$, which requires $h \ge 3$. Note that $\cF_\ell$ has at least $N^{h-1}$ elements. We now have the following two key claims.

\begin{claim}\label{claim:proj-dist}
For any distinct $C_{\ell, j_1, \hdots, j_{h-1}}, C_{\ell, j'_1, \hdots, j'_{h-1}} \in \cF_\ell$, we have that the vectors $\Sigma_H(C_{\ell, j_1, \hdots, j_{h-1}})$ and $\Sigma_H(C_{\ell, j'_1, \hdots, j'_{h-1}}\change{)}$ are linearly independent (i.e., projectively distinct).
\end{claim}
\begin{proof}
Let $E' := C_{\ell, j_1, \hdots, j_{h-1}} \cup C_{\ell, j'_1, \hdots, j'_{h-1}}$. By Proposition~\ref{prop:cycle-ind-h'}, it suffices to prove there is an acyclic subgraph $T \subseteq E'$ for which $|E' \setminus T| \le h$ and that $C_{\ell, j_1, \hdots, j_{h-1}} \setminus C_{\ell, j'_1, \hdots, j'_{h-1}}$ and $C_{\ell, j_1, \hdots, j_{h-1}} \setminus C_{\ell, j'_1, \hdots, j'_{h-1}}$ are both nonempty. The nonemptyness condition follows from the fact that the two cycles have the exact same number of edges but are distinct.

For the acyclic subgraph condition. Let $J = \{j_1, \hdots, j_{h-1}\} \cap \{j'_1, \hdots, j'_{h-1}\}$. One can count that $|E'| = 4(h-1) - 2|J|$, since every $j \in J$ appears in the same edges in both cycles. Let $k =  |\{j_1, \hdots, j_{h-1}\} \cup \{j'_1, \hdots, j'_{h-1}\}| = 2h-2 - |J|$.

Since $I_\ell$ has $h-1$ left vertices and $E'$ is a connected graph, $E'$ has a spanning tree $T$ with $k + h - 2 = 3h - 4 - |J|$ edges. Thus, $|E' \setminus T| = |E'| - |T| = (4h - 4 - 2|J|) - (3h - 4 - |J|) = h - |J| \le h$, as desired.
\end{proof}

\begin{claim}\label{claim:lin-ind}
For any choice of $C_1 \in \cF_1, \hdots, C_h \in \cF_h$, we have that $\Sigma_H(C_1), \hdots, \Sigma_H(C_h)$ are linearly independent.
\end{claim}
\begin{proof}
By Proposition~\ref{prop:cycle-ind}, it suffices to prove that there is an acyclic graph $T \subseteq C_1 \cup \cdots \cup C_h$ such that $|C_\ell \setminus T| = 1$ for all $\ell \in [h]$. By our choice of $I_1, \hdots, I_h$ and $J_1, \hdots, J_{h(h-1)}$, these $h$ cycles have disjoint edge sets. Therefore, if we delete from $C_1 \cup \cdots \cup C_h$ any arbitrary edge of each $C_\ell$ for all $\ell \in [h]$, we are left with an acyclic graph, as desired.
\end{proof}

We can thus, define $A_\ell := \{\Sigma_H(C) : C \in \cF_\ell\}$ for all $\ell \in [h]$. By Claim~\ref{claim:proj-dist}, each $A_\ell$ has size at least $N^{\change{h-1}}$ and its elements are projectively distinct. Furthermore, by Claim~\ref{claim:lin-ind} any choice of $a_1 \in A_1, \hdots, a_h \in A_h$ are linearly independent, and thus $A_1, \hdots, A_h$ are disjoint. We can apply Proposition~\ref{prop:q-ind} to get $q \ge \frac{N^{h-1}}{h-1}-4 = \Omega(n^{h-1} / h^{2h-1})$, as desired.
\end{proof}

\subsection{Near-monotonicity in \texorpdfstring{$h$}{h}}

\change{A general question one may ask about $q(m,n,a,b,h)$ is the dependence on the field size as a function of the number of global parity checks $h$. Intuitively, demanding more global parity checks appears to be a stronger property of the code and thus the field size should grow as a function of $h$. However, rather surprisingly, no such result is known. The primary issue is that if $h > h'$, an MR grid code with parameters $(m,n,a,b,h)$ does not immediately imply an MR grid code with parameters $(m,n,a,b,h')$ as the former code has smaller dimension than the latter, and it is not clear how to make up for the deficit in dimension.}

As mentioned in the introduction, Gopalan et al.~\cite{Gopalan2016} (as stated by Kane et al.~\cite{kane2019independence}) states that
\begin{align*}
q(m,n,a,b,h) \ge q(\min(m-a+1,n-b+1),\min(m-a+1,n-b+1),1,1,1)^{1/h}.
\end{align*}
When $m$ is a constant and $n$ is growing, this lower bound is a constant. We greatly improve on this lower bound with the following result.

\begin{theorem}\label{thm:monotone-h}
For any \change{$m \ge 3$, $h \ge h' \ge 1$, and $n \ge h-h'+2$}, we have that
\[
    q(m, n, 1, 1, h) \ge q(m-2, n-h+h'-1, 1, 1, h')
\]
\end{theorem}
For example, when $m = 4$, $h=2$ and $h'=1$, when we combine Theorem~\ref{thm:monotone-h} with Corollary~\ref{cor:mge2lb}, we have that \[q(4, n, 1, 1, 2) \ge n-2,\] the first superconstant lower bound for these parameters.

\begin{remark}
In the Appendix, we prove a variant of Theorem~\ref{thm:monotone-h} which gives improved lower bounds when $m$ is superconstant. See Theorem~\ref{thm:box-bound}.
\end{remark}

\begin{proof}[Proof of Theorem~\ref{thm:monotone-h}]
Let $H \in \F_q^{(m+n+h-1) \times mn}$ be the parity check matrix of a MR $(m,n,a=1,b=1,h)$ grid code with corresponding parity check vectors $c^{i,j} \in \F_q^h$ for all $(i,j) \in [m] \times [n]$. It suffices to construct the parity check matrix of a MR $(m-2,n-(h-h'+1),a=1,b=1,h=h')$ grid code over the same field. In particular, we construct $d^{i,j} \in \F_q^{h'}$ for $(i,j) \in [m-2] \times [n-(h-h'+1)]$ with the MR property.

For $\ell \in [h-h']$, let $C_\ell \subset [m] \times [n]$ be the following cycle
\[
    C_\ell := \{(m-1, n-h+h'), (m-1, n-h+h'+\ell), (m, n-h+h'+\ell), (m, n-h+h')\}
\]
Note that $C_1 \cup \cdots \cup C_{h-h'} = \{m-1,m\} \times \{n-h+h', \hdots, n\}$. Consider the tree $T = \{m-1\} \times \{n-h+h', \hdots, n\} \cup \{(m, n-h+h')\}$, we have for all $\ell \in [h-h']$ then, $C_\ell \setminus T = \{(m, n-h+h'+\ell)\}$ has size exactly one and is distinct from $C_{\ell'} \setminus T$ for all other $\ell' \in [h-h']$. Thus by Proposition~\ref{prop:cycle-ind}, we have that $\Sigma_H(C_1), \hdots, \Sigma_H(C_{h-h'}) \in \F_q^h$ are linearly independent.

Thus, there exists a linear map $\psi : \F_q^h \to \F_q^{h'}$ such that $\ker \psi = \operatorname{span}(\Sigma_H(C_1), \hdots, \Sigma_H(C_{h-h'}))$. Let $m' := m-2$ and $n' := n-h+h'-1$. Define $d^{i,j} \in \F^{h'}_q$ for all $(i,j) \in [m'] \times [n']$ as follows:
\begin{align}
    d^{i,j} := \psi(c^{i,j}).\label{eq:d}
\end{align}
Let $H' \in \F_q^{(m'+n') \times m'n'}$ be the parity check matrix for an $(m', n', a=1, b=1, h=h')$ grid code with global parity checks corresponding to $\{d^{i,j} : (i,j) \in [m'] \times [n']\}$.

We claim that $H'$ is MR. Consider any $h'$ cycles $C'_1, \hdots, C'_{h'} \subseteq [m'] \times [n']$ for which there is a spanning tree $T' \subseteq [m'] \times [n']$ such that $C'_{\ell'} \setminus T'$ \change{has a unique single edge} for all $\ell' \in [h']$.  By Proposition~\ref{prop:cycle-ind}, \change{it suffices to show} that $\Sigma_{H'}(C'_1), \hdots, \Sigma_{H'}(C'_{h'})$ are linearly independent. \change{B}y definition of $d^{i,j}$, we know that 
\[
\Sigma_{H'}(C'_i) = \psi(\Sigma_{H}(C'_i))
\] for $i\in [h']$.
Thus, since $\ker \psi = \operatorname{span}(\Sigma_H(C_1), \hdots, \Sigma_H(C_{h-h'}))$, it suffices to prove that $\Sigma_H(C_1), \hdots,$ $\Sigma_H(C_{h-h'}),$ $\Sigma_H(C'_1), \hdots, \Sigma_H(C'_{h'}) \in \F_q^h$ are linearly independent. Since $T$ and $T'$ are on disjoint sets of vertices, we have that $T \cup T'$ is acyclic and each of
\[
C_1 \setminus (T \cup T'), \cdots C_{h-h'} \setminus (T \cup T'), C'_1 \setminus (T \cup T'), \hdots, C'_{h'} \setminus (T \cup T')
\]
has cardinality one and is distinct. Thus by Proposition~\ref{prop:cycle-ind}, we have that $\Sigma_H(C_1), \hdots, \Sigma_H(C_{h-h'}),$ $\Sigma_H(C'_1), \hdots, \Sigma_H(C'_{h'})$ are linearly independent because $H$ is MR, as desired.
\end{proof}
By combining Theorem~\ref{thm:lb-general} with $h' = \min(\lfloor\sqrt{m-2}\rfloor+1, h)$ and Theorem~\ref{thm:monotone-h} (with Corollary~\ref{cor:mge2lb} for $h=2$), we get the following result.

\begin{corollary}\label{cor:combine}
For $n,m,h \ge 2$, $q(m,n,1,1,h) = \Omega_{m,h}(n^{\min(\lfloor\sqrt{m-2}\rfloor,h-1)})$.
\end{corollary}

\section{Conclusion}

In this paper, we proved a number of new results on the field size of MR $(m,n,1,1,h)$ grid codes. We hope the concrete constructions we present in this paper could potentially be of use in the future design of distributed storage systems.

We leave the following as an open question: what is the optimal field size of for an MR $(m,n,1,1,h)$ grid code when $m$ and $h$ are constant? In particular, when $h=1$, we give an upper bound of $n^{m-1}$, but we can only prove a lower bound of $n$. Which of these two bounds is closer to the correct answer? Decreasing the exponent (without increasing multiplicative factors) could also perhaps lead to more scenarios in which these codes are practically useful.

We also note that many open questions remain even in the regime for which $m$ grows as a function of $n$, particularly when $h \ge 2$. See the Appendix for more details.

\section*{Acknowledgments}

We thank Venkatesan Guruswami and Sergey Yekhanin for valuable comments on this manuscript. \change{We also thank anonymous reviewers for numerous suggestions improving the quality and correctness of this manuscript.}

\bibliographystyle{IEEEtran}
\bibliography{references}
\appendix

Although our paper is focused on the regime in which $m$ is constant, our methods also produce some novel results in the regime in which $m$ is growing as a function of $n$ (including the well-studied case of $m=n$). We prove the following results.

First, the best lower bound on $q(n, n, 1, 1, h)$ due to Kane et al.~\cite{kane2019independence} is of the form $2^{\Omega(n/h)}$. Using the same proof strategy as Theorem~\ref{thm:monotone-h}, we can improve this bound to $2^{\Omega(n - \sqrt{h})}$ for any $h \ge 2$.

\begin{theorem}\label{thm:box-bound}
For any $m, n, h \ge 1$, let $h_1$ and $h_2$ positive integers such that $(h_1-1)(h_2-1) \ge h-h'$. We then have that
\[
    q(m, n, 1, 1, h) \ge q(m-h_1, n-h_2, 1, 1, h').
\]
In particular by setting $h'=1$ and $h_1 = h_2 = \lceil \sqrt{h-1}\rceil + 1$, we have that $q(n, n, 1, 1, h) \ge 2^{\Omega(n - \sqrt{h})}.$ 
\end{theorem}

This loss of $\sqrt{h}$ (essentially) tightly matches the extreme parameter setting of $h=(m-1)(n-1)$ in which case a trivial construction exists over field size $2$.

\begin{proof}[\change{Proof of Theorem~\ref{thm:box-bound}}]
Since the proof of Theorem~\ref{thm:box-bound} is very similar to that of Theorem~\ref{thm:monotone-h}, we only discuss the key changes to the proof.

Let $H$ be the parity check code of an MR $(m, n, 1, 1, h)$ grid code. Let $I := \{m-h_1+1, \hdots, m\}$ and $J := \{n-h_2+1, \hdots, n\}$. Let $T$ be a spanning tree of $I \times J$. Note that $|T| = h_1+h_2-1$, so $(I \times J) \setminus T$ has cardinality $(h_1-1)(h_2-1) \ge h-h'$. Thus, we can pick $h-h'$ cycles $C_1, \hdots, C_{h-h'} \subseteq I \times J$ such that for all $\ell \in [h-h']$, $C_\ell \setminus T$ has size one and a single edge. Thus by Proposition~\ref{prop:cycle-ind}, we have that $\Sigma_H(C_1), \hdots, \Sigma_H(C_{h-h'}) \in \F_q^h$ are linearly independent.

Let $m' := m-h_1$ and $n' := n-h_2$. We can then define $d^{i,j} \in \F^{h'}_q$ for $(i,j) \in [m'] \times [n']$ precisely as in (\ref{eq:d}). To show that the $d^{i,j}$'s form the global parity check of a MR $(m', n', 1, 1, h')$ grid code, consider any $h'$ cycles $C'_1, \hdots, C'_{h'} \subseteq [m'] \times [n']$ for which there is a spanning tree $T' \subseteq [m] \times [n]$ such that $C'_{\ell'} \setminus T$ is a unique singleton for all $\ell' \in [h']$.  It then suffices to prove that  $\Sigma_H(C_1), \hdots, \Sigma_H(C_{h-h'}),$ $\Sigma_H(C'_1), \hdots, \Sigma_H(C'_{h'}) \in \F_q^h$ are linearly independent. This follows from noting that deleting $T \cup T'$ leaves a single distinct edge form each of the $h$ cycles.

Thus, an MR $(m', n', 1, 1, h')$ grid code exists over $\F_q$, implying the lower bound. The asymptotic bound $q(n, n, 1, 1, h) \ge 2^{\Omega(n - \sqrt{h})}$ immediately follows by composing this reduction with the bound of Kane et al.~\cite{kane2019independence}.
\end{proof}

We also give a new construction. In the setting of $h=1$, recall that Kane et al.~\cite{kane2019independence} shows that $q(m, n, 1, 1, 1) \le 8^{\max(m,n)}$, when $\max(m, n)$ is a power of two. Furthermore, our Theorem~\ref{thm:h1m-1} shows that $q(m, n, 1, 1, 1) \le n^{m-1}$, when $n$ is a power of two. We interpolate these constructions as follows.

\begin{theorem}\label{thm:h1-bootstrap}
Assume that $m \le n$ are powers of two, and there exists an MR $(m, m, a=1, b=1, h=1)$ grid code over field size $q_m$, where $q_m$ is a power of two. Then, there exists an MR $(m, n, a=1, b=1, h=1)$ grid code over field size $q = (n/m)^{m-1} q_m$. In particular, $q(m, n, 1, 1, 1) \le 8(8n/m)^{m-1}$.
\end{theorem}

\begin{proof}[\change{Proof of Theorem~\ref{thm:h1-bootstrap}}]
This construction is partially inspired by the recursive construction of Kane et al.~\cite{kane2019independence}.

Assume $q_m = 2^k$. Let $\gamma : [m] \times [m] \to \F_2^k$ be the labeling corresponding to an MR $(m, m, a=1, b=1, h=1)$ grid code. Let $b = \log_2(m / n)$. We construct an extended map $\gamma' : [m] \times [n] \to \F_2^{k + b (m-1)}$ as follows.

Consider any $i \in [m]$ and $j \in [n]$, where we express $j = j'm + j''$ with $j' \in \{0, 1, \hdots, n/m - 1\}$ and $j'' \in [m]$. We set the first $k$ bits of $\gamma'(i, j)$ are equal to $\gamma(i, j'')$. For the remaining $b(m-1) = \log_2(n/m)(m-1)$ bits, partition these into $m-1$ blocks of length $b$. If $i \in [m-1]$, set the $i$th block equal to the binary representation of $j'$, with all other blocks equal to zero. If $i = m$, set all $m-1$ blocks equal to zero.

We claim this construction is MR $(m, n, a=1, b=1, h=1)$. To see why, assume for sake of contradiction that $C \subseteq [m] \times [n]$ is a simple cycle for which $\gamma'(C) = 0$. Consider any $i \in [m-1]$, for which there exist distinct $j_1, j_2 \in [n]$ for which $(i, j_1), (i, j_2) \in C$. Further assume that $j_1 = j'_1m + j_1''$ and $j_2 = j'_2m + j_2''$. Then, observe that $\gamma'(i, j_1)$ and $\gamma'(i, j_2)$ are the only edges of the cycle for which their labels which give the $i$th binary block a nonzero label. Thus, $\gamma(i, j_1)$ and $\gamma(i, j_2)$ are equal in that block. Therefore, the $j'_1 = j'_2$. Repeating this argument for all other left vertices in $C$ (except possibly $m$), we can deduce that $j'_1 = j'_2$ for all right vertices $j_1$ and $j_2$ of $C$. Thus, $C$ is a cycle contained within one of the $[m] \times [m]$ blocks of $[m] \times [n]$. However, the first $k$ bits of $\gamma'$ being $\gamma$ ensure that any such cycle has nonzero sum, a contradiction. Therefore, the construction is MR.
\end{proof}

A major question in this regime that we leave open is whether the upper bound of $n^{O(n)}$ (due to Gopalan et al.~\cite{gopalan2014explicit,Gopalan2016}) for the setting in which $m=n$ and $h \ge 2$ can be improved upon.

\end{document}